\documentclass[10pt,conference,letter]{IEEEtran}

\usepackage{times}

\usepackage{amssymb}
\usepackage{amsthm}
\usepackage{amsxtra}
\usepackage{amsmath}

\usepackage{verbatim}
\usepackage{epsfig}
\usepackage{setspace}
\usepackage{caption}
\usepackage{subcaption}
\usepackage{breqn}

\allowdisplaybreaks

\newtheorem{theorem}{Theorem}
\newtheorem{definition}{Definition}
\newtheorem{lemma}[theorem]{Lemma}
\newtheorem{proposition}[theorem]{Proposition}

\newtheorem{example}{Example}[section]

\def\CI{\text{CI}}
\def\LCI{\text{LCI}}
\def\CO{\text{CO}}
\def\SK{\text{SK}}
\def\rank{\text{rank}}

\def\cB{\mathcal{B}}
\def\cE{\mathcal{E}}
\def\cF{\mathcal{F}}
\def\cG{\mathcal{G}}
\def\cM{\mathcal{M}}
\def\cV{\mathcal{V}}
\def\cX{\mathcal{X}}

\def\F{\mathbb{F}}
\def\N{\mathbb{N}}

\title{On the Communication Complexity of Secret Key Generation in the Multiterminal Source Model}
\author{
\IEEEauthorblockN{Manuj Mukherjee$^\dag$} \and \IEEEauthorblockN{Navin Kashyap$^\dag$}
}

\begin{document}

\maketitle

\renewcommand{\thefootnote}{}
\footnotetext{
\noindent $^\dag$M.\ Mukherjee and N.\ Kashyap are with the 
Department of Electrical Communication Engineering, 
Indian Institute of Science, Bangalore. Email: \{manuj,nkashyap\}@ece.iisc.ernet.in.

\smallskip
}

\begin{abstract}
Communication complexity refers to the minimum rate of public communication required for generating a maximal-rate secret key (SK) in the multiterminal source model of Csisz{\'a}r and Narayan. Tyagi recently characterized this communication complexity for a two-terminal system. We extend the ideas in Tyagi's work to derive a lower bound on communication complexity in the general multiterminal setting. In the important special case of the complete graph pairwise independent network (PIN) model, our bound allows us to determine the exact linear communication complexity, i.e., the communication complexity when the communication and SK are restricted to be linear functions of the randomness available at the terminals.
\end{abstract}

\renewcommand{\thefootnote}{\arabic{footnote}}

\section{Introduction} \label{sec:intro}
Csisz{\'a}r and Narayan \cite{CN04} introduced the problem of secret key (SK) generation within the multiterminal source model. In this model, there are multiple terminals, each of which observes a distinct component of a source of correlated randomness. The terminals must agree on a shared SK by communicating over a noiseless public channel. This key is to be protected from a passive eavesdropper having access to the public communication. Various equivalent characterizations of the \emph{SK capacity}, i.e., the supremum of the rates of SKs that can be generated within this model, are now known \cite{CN04}, \cite{CZ10}, \cite{NN10}. Proofs of achievability of the SK capacity typically involve communication protocols that enable ``omniscience'' at all terminals, which means that the communication over the public channel allows each terminal to recover the observations of all the other terminals. On the other hand, it is known (see remark following Theorem~1 in \cite{CN04}) that omniscience is not necessary for maximal-rate SK generation. Thus, communication enabling omniscience may be wasteful in terms of rate. In this paper, we are concerned with the problem of determining the \emph{communication complexity} of achieving SK capacity, i.e., the minimum rate of communication required to generate a maximal-rate SK. 

In the case when there are only two terminals in the model, Tyagi gave an exact characterization of the communication complexity \cite[Theorem~3]{Tyagi13} in terms of the minimum rate of an ``interactive common information'', a type of Wyner common information \cite{Wyner75}. We extend the main ideas of Tyagi's work to the general setting of $m \ge 2$ terminals, and obtain a lower bound on the communication complexity of SK capacity. While we can show that our bound is always non-negative, evaluating the bound seems to be difficult even in well-studied special cases like the pairwise independent network (PIN) model of \cite{NN10}. 

In the PIN model, Nitinawarat and Narayan \cite{NN10} have shown that a maximal-rate SK can be generated by a protocol in which the public communication and the SK generated are both linear functions of the observations of the terminals\footnote{Indeed, the protocol given in the proof of \cite[Theorem 1]{NN10} to obtain a maximal-rate SK uses public communication that is a linear function of the terminals' observations. Though it is not explicitly stated in the proof, it is easy to see that the SK function can also be chosen to be linear.\label{fn:linSK}}. We can then define the \emph{linear communication complexity} of achieving SK capacity as the minimum rate of communication required when the communication and the SK are restricted to be linear functions of the observations. An appropriately modified version of our lower bound applies in this linear setting. We are able to explicitly evaluate our bound in the particular case of the complete graph PIN model. The SK-capacity-achieving protocol in the proof of \cite[Theorem~1]{NN10} uses a linear communication that enables omniscience at all terminals; the rate of this communication is an upper bound on the linear communication complexity. For the complete graph PIN model on $m \ge 2$ terminals, our lower bound meets this upper bound: the linear communication complexity in this case equals $m(m-2)/2$. This exact result in an important special case is a testament to the power of our lower bounding method.

The rest of the paper is structured as follows. Section~\ref{sec:prelim} presents the required definitions and notation. Section~\ref{sec:main} describes our lower bound on the communication complexity of achieving SK capacity. In Section~\ref{sec:linear}, we adapt our bound to the linear setting and evaluate it for the complete graph PIN model. The paper concludes with some remarks in Section~\ref{sec:conc}.

\section{Preliminaries} \label{sec:prelim}
 Throughout, we use $\N$ to denote the set of positive integers. Consider a set of $m$ terminals denoted by $\mathcal{M}=\{1,2, \ldots, m\}$. Each terminal $i \in \mathcal{M}$ observes $n$ i.i.d.\ repetitions of the random variable $X_i$ taking values in the finite set $\mathcal{X}_i$. The $n$ i.i.d.\ copies of the random variable are denoted by $X_i^n$. For any subset $A\subseteq \mathcal{M}$, $X_A$ and $X_A^n$ denote the collections of random variables $(X_i:i \in A)$ and $(X_i^n: i \in A)$, respectively. The terminals communicate through a noiseless public channel, 
any communication sent through which is accessible to all terminals and to potential eavesdroppers as well.
An \emph{$r$-interactive communication} is a communication $\textbf{f}=(f_1,f_2,\cdots,f_r)$ consisting of $r$ transmissions. Any transmission sent by the $i$th terminal is a deterministic function of $X_i^n$ and all the previous communication, i.e., if terminal $i$ transmits $f_j$, then $f_j$ is a function only of $X_i^n$ and $f_1,\ldots,f_{j-1}$.   We denote the random variable associated with $\textbf{f}$ by $\textbf{F}$; the support of $\textbf{F}$ is a finite set $\cF$. The rate of the communication $\textbf{F}$ is defined as $\frac{1}{n}\log|\cF|$. Note that $\textbf{f}$, $\textbf{F}$ and $\cF$ implicitly depend on $n$.

\begin{definition}
\label{def:CR}
A \emph{common randomness (CR)} obtained from an $r$-interactive communication $\textbf{F}$ is a sequence of random variables $\textbf{J}^{(n)}$, $n\in\N$, which are functions of $X_{\mathcal{M}}^n$, such that for any $0<\epsilon<1$ and for all sufficiently large $n$, there exist $J_i=J_i(X_i^n,\textbf{F})$, $i = 1,2,\ldots,m$, satisfying $Pr\{J_1=J_2=\cdots=J_m=\textbf{J}^{(n)}\}\geq 1-\epsilon$.
\end{definition}

\begin{definition}
\label{def:SK}
A real number $R\geq 0$ is an \emph{achievable SK rate} if there exists a CR $\textbf{K}^{(n)}$, $n \in \N$, obtained from an $r$-interactive communication $\textbf{F}$ satisfying, for any $\epsilon > 0$ and for all sufficiently large $n$, $I(\textbf{K}^{(n)};\textbf{F})\leq \epsilon$ and $\frac{1}{n}H(\textbf{K}^{(n)}) \geq R-\epsilon$.\footnote{Usually, an additional requirement that $\textbf{K}^{(n)}$ be almost uniformly distributed over its alphabet $\mathcal{K}^{(n)}$, i.e., $H(\textbf{K}^{(n)}) \ge \log |\mathcal{K}^{(n)}| - \epsilon$, is also included in this definition. However, this can always be dropped without affecting SK capacity or communication rates --- see e.g., \cite[p.~3976]{GA10}.} The \emph{SK capacity} is defined to be the supremum among all achievable rates.  The CR $\textbf{K}^{(n)}$ is called a \emph{secret key (SK)}. 
\end{definition}

From now on, we will drop the superscript $(n)$ from both $\textbf{J}^{(n)}$ and $\textbf{K}^{(n)}$ to keep the notation simple. 

The SK capacity can be expressed as \cite[Section~V]{CN04}, \cite{CZ10}
\begin{equation}
\textbf{I}(X_{\mathcal{M}}) \triangleq H(X_{\mathcal{M}})-\max_{\lambda \in \Lambda} \sum_{B \in \mathcal{B}} \lambda_B H(X_B|X_{B^c}) \label{skcapacity}
\end{equation}
where $\mathcal{B}$ is the set of non-empty, proper subsets of $\mathcal{M}$ and $\lambda=(\lambda_B: B\in \mathcal{B})\in \Lambda$ iff $\lambda_B\geq 0$ for all $B\in \mathcal{B}$ and for all $i\in \mathcal{M}$, $\sum_{B:i\in B}\lambda_B=1$. It is a fact that $\textbf{I}(X_{\mathcal{M}}) \ge 0$ \cite[Proposition~II]{MT10}. From now on, we will denote the optimal $\lambda\in\Lambda$ for the linear program in (\ref{skcapacity}) by $\lambda^*$.

We are now in a position to make the notion of communication complexity rigorous.

\begin{definition}
\label{def:RSKr}
Let $r\geq m$ be fixed. A real number $R\geq 0$ is said to be an \emph{achievable rate of $r$-interactive communication for maximal-rate SK} if for all $\epsilon > 0$ and for all sufficiently large $n$, there exist \emph{(i)}~an $r$-interactive communication $\textbf{F}$ satisfying $\frac{1}{n}\log|\cF| \; \leq R+\epsilon$, and \emph{(ii)}~an SK $\textbf{K}$ obtained from $\textbf{F}$ such that $\frac{1}{n}H(\textbf{K})\geq \textbf{I}(X_{\mathcal{M}})-\epsilon$.

We denote the infimum among all such achievable rates by $R_{\SK}^r$.
\end{definition}

The $r \ge m$ condition in the above definition requires a note of explanation. The proof of Theorem~1 in \cite{CN04} shows that there exists an $m$-interactive communication $\textbf{F}$ that enables omniscience at all terminals and from which a maximal-rate SK can be obtained. Thus, for $r \ge m$, we have $R_{\SK}^r < \infty$. Another point to be noted is that $R_{\SK}^r$ is a non-increasing function of $r$, since any rate achievable with $r$ transmissions is also achievable with $r+1$ transmissions (by, say, keeping the last transmission silent). Hence, we can define 
\begin{equation}
R_{\SK} \triangleq \lim_{r \to \infty} R_{\SK}^r
\label{def:RSK}
\end{equation} 
to be the communication complexity of generating a maximal-rate SK. 


Tyagi gave a characterization of $R_{\SK}$ in the case of a two-terminal model \cite[Theorem~3]{Tyagi13}.\footnote{It should be clarified that Tyagi's characterization works only for ``weak'' SKs, which are defined as in our Definition~\ref{def:SK}, except that the condition $I(\textbf{K};\textbf{F}) \le \epsilon$ is weakened to $\frac{1}{n} I(\textbf{K};\textbf{F}) \le \epsilon$.  Using our definitions, Tyagi's arguments would only yield a two-terminal analogue of our Theorem~\ref{th:main}.\label{fn:weakSK}} The key to his characterization was the observation that conditioned on a maximal-rate SK $\textbf{K}$ and the communication $\textbf{F}$ from which $\textbf{K}$ is extracted, the observations of the two terminals are ``almost'' independent: $\frac{1}{n}I(X_1^n;X_2^n \mid \textbf{K},\textbf{F}) \to 0$ as $n \to \infty$. Thus, the pair $(\textbf{K},\textbf{F})$ is a Wyner common information \cite{Wyner75} for the randomness at the terminals. Tyagi used the term ``interactive common information'' to denote any Wyner common information that consisted of a CR along with the communication achieving it. We now extend these definitions to the multiterminal setting.

We will need the following extension of the definition of $\textbf{I}(X_{\mathcal{M}})$ given in \eqref{skcapacity}: for any random variable $\textbf{L}$, and any $n \in \N$, we define
\begin{equation}
\textbf{I}(X_{\mathcal{M}}^n | \textbf{L}) \triangleq H(X_{\mathcal{M}}^n | \textbf{L})-\sum_{B \in \mathcal{B}} \lambda^*_B H(X_B^n|X_{B^c}^n,\textbf{L}), \label{cmi}
\end{equation}
where $\lambda^* = (\lambda_B^*: B \in \mathcal{B})$ is the optimal $\lambda \in \Lambda$ for the linear program in the definition of $\textbf{I}(X_{\mathcal{M}})$ in \eqref{skcapacity}. It follows from Proposition~II in \cite{MT10} that $\textbf{I}(X_{\mathcal{M}}^n | \textbf{L}) \ge 0$. Also, note that $\textbf{I}(X_{\cM}^n) = n\textbf{I}(X_{\cM})$. 

\begin{definition}
\label{def:CI}
A \emph{(multiterminal) Wyner common information ($\CI_W$)} for $X_{\mathcal{M}}$ is a sequence of finite-valued functions $\textbf{L}^{(n)}=\textbf{L}^{(n)}(X_{\mathcal{M}}^n)$ such that $\frac{1}{n}\textbf{I}(X_{\mathcal{M}}^n|\textbf{L}^{(n)}) \to 0$ as $n \to \infty$. An \emph{$r$-interactive common information ($\CI^r$)} for $X_{\mathcal{M}}$ is a Wyner common information of the form $\textbf{L}^{(n)} = (\textbf{J},\textbf{F})$, where $\textbf{F}$ is an $r$-interactive communication and $\textbf{J}$ is a CR obtained from $\textbf{F}$. 
\end{definition}

Again, we shall drop the superscript $(n)$ from $\textbf{L}^{(n)}$ for notational simplicity. Wyner common informations $\textbf{L}$ do exist: for example, the identity map $\textbf{L}=X_{\mathcal{M}}^n$ is a $\CI_W$. 
To see that $\CI^r$s $(\textbf{J},\textbf{F})$ also exist, observe that $\textbf{J}=X_{\mathcal{M}}^n$ and a communication $\textbf{F}$ enabling omniscience constitute a $\CI_W$, and hence, a $\CI^r$.
The proof of \cite[Theorem 1]{CN04} shows that there exists a communication of $m$ transmissions that enables omniscience. It follows that a $\CI^r$ exists for any $r \ge m$. 


\begin{definition}
\label{def:CIrate}
A real number $R\geq 0$ is an \emph{achievable $\CI_W$ (resp.\ $\CI^r$) rate} if there exists a $\CI_W$ $\textbf{L}$ (resp.\ a $\CI^r$ $\textbf{L} = (\textbf{J},\textbf{F})$) such that for all $\epsilon > 0$, we have
$\frac{1}{n}H(\textbf{L})\leq R+\epsilon$ for all sufficiently large $n$. 

We denote the infimum among all achievable $\CI_W$ rates by $\CI_W(X_{\mathcal{M}})$. For $r \ge m$, we denote the infimum among all achievable $\CI^r$ rates by $\CI^r(X_{\mathcal{M}})$. 
\end{definition}


The explanation for the $r\geq m$ condition in the definition of $\CI^r(X_{\mathcal{M}})$ above is similar to that given after Definition~\ref{def:RSKr}. To ensure that $\CI^r(X_{\mathcal{M}}) < \infty$, the existence of at least one $\CI^r$ pair $(\textbf{J},\textbf{F})$ is needed, and as observed earlier, this is guaranteed when $r \ge m$. The rate achieved by this guaranteed $\CI^r$ pair $(\textbf{J},\textbf{F})$ is $H(X_{\mathcal{M}})$. 

Furthermore, analogous to $R_{\SK}^r$, $\CI^r(X_{\mathcal{M}})$ is a non-increasing function of $r$. Hence, we can define $\CI(X_{\mathcal{M}}) \triangleq \lim_{r \to \infty} \CI^r(X_{\mathcal{M}})$. The proposition below records the relationships between some of the information-theoretic quantities defined so far. 

\begin{proposition} For all $r \ge m$, we have 
$H(X_{\mathcal{M}}) \ge \CI^r(X_{\mathcal{M}})\geq \CI(X_{\mathcal{M}})\geq \CI_W(X_{\mathcal{M}})\geq \textbf{I}(X_{\mathcal{M}})$.
\label{prop:ineqs}
\end{proposition}
\begin{proof}
The first inequality is due to the fact that for any $r \ge m$, there exists a $\CI^r$ of rate $H(X_{\mathcal{M}})$. The second inequality is trivial. The third follows from the fact that a $\CI^r$ is a special type of $\CI_W$, so that $\CI^r(X_{\mathcal{M}}) \ge \CI_W(X_{\mathcal{M}})$.

For the last inequality, we start by observing that for any function $\textbf{L}$ of $X_{\mathcal{M}}^n$, we have 
\begin{align*}
\textbf{I}(X_{\mathcal{M}}^n)-\textbf{I}(X_{\mathcal{M}}^n|\textbf{L}) 
    &=I(X_{\mathcal{M}}^n;\textbf{L})-\sum_{B\in\mathcal{B}}\lambda_B^*I(X_B^n;\textbf{L}|X_{B^c}^n) \\
   &=H(\textbf{L})-\sum_{B\in\mathcal{B}}\lambda_B^*H(\textbf{L}|X_{B^c}^n) 
\end{align*}
Hence, 
\begin{equation}
\frac{1}{n}H(\textbf{L})\geq\textbf{I}(X_{\mathcal{M}})-\frac{1}{n}\textbf{I}(X_{\mathcal{M}}^n|\textbf{L}).
\label{eq:a}
\end{equation} 
Now, if $\textbf{L}$ is any $\CI_W$ of rate $R$, then by Definitions~\ref{def:CI} and \ref{def:CIrate}, for every $\epsilon>0$, we have $\frac{1}{n}H(\textbf{L}) \le R + \epsilon$ and $\frac{1}{n}\textbf{I}(X_{\mathcal{M}}^n|\textbf{L})\leq \epsilon$ for all sufficiently large $n$. Thus, in conjunction with \eqref{eq:a}, we have $R + \epsilon \ge \frac{1}{n}H(\textbf{L}) \ge \textbf{I}(X_{\mathcal{M}})-\epsilon$ for all sufficiently large $n$. In particular, $R+\epsilon \ge \textbf{I}(X_{\mathcal{M}})-\epsilon$ holds for any $\epsilon > 0$, from which we infer that $R \ge \textbf{I}(X_{\mathcal{M}})$. The inequality $\CI_W(X_{\mathcal{M}}) \ge \textbf{I}(X_{\mathcal{M}})$ now follows.
\end{proof}


Finally, analogous to Definition~\ref{def:RSKr}, we have a definition of achievable rate of $r$-interactive communication required to get a $\CI^r$. 

\begin{definition}
\label{def:RCI}
Let $r\geq m$ be fixed. A real number $R\geq 0$ is said to be an \emph{achievable rate of $r$-interactive communication for $\CI^r$} if for all $\epsilon > 0$ and for all sufficiently large $n$, there exist \emph{(i)}~an $r$-interactive communication $\textbf{F}$ satisfying $\frac{1}{n}\log|\cF| \; \leq R+\epsilon$, and \emph{(ii)}~a CR $\textbf{J}$ such that $\textbf{L}=(\textbf{J},\textbf{F})$ is a $\CI^r$.

We denote the infimum among all such achievable rates by $R_{\CI}^r$.
\end{definition}

As was the case with $R_{\SK}^r$, we observe that $R_{\CI}^r$ is a non-increasing sequence in $r$, bounded below by $0$. Thus, we can define $R_{\CI} \triangleq\lim_{r\to\infty}R_{\CI}^r$. The main theorem of our paper, stated in the next section, gives a lower bound on the communication complexity $R_{\SK}$, expressed in terms of the quantities defined in Definitions~\ref{def:CI}--\ref{def:RCI}.

\section{Lower Bound on $R_{\SK}$} \label{sec:main}

The goal of this section is to state and prove the main result of this paper, which partially extends Tyagi's two-terminal result \cite[Theorem~3]{Tyagi13} to the multiterminal setting.

\begin{theorem}
For all $r\geq m$, we have 
$$
R_{\SK}^r\geq R_{\CI}^r\geq \CI^r(X_{\mathcal{M}})-\textbf{I}(X_{\mathcal{M}}).
$$
Hence, by letting $r\to \infty$,
$$
R_{\SK}\geq R_{\CI}\geq \CI(X_{\mathcal{M}})-\textbf{I}(X_{\mathcal{M}}).
$$
\label{th:main}
\end{theorem}
\vspace*{-10pt} 
By Proposition~\ref{prop:ineqs}, the lower bounds above are non-negative.

The ideas in our proof of Theorem~\ref{th:main} may be viewed as a natural extension of those in the proof of \cite[Theorem 3]{Tyagi13}. We start with three preliminary lemmas. In all that follows, $\lambda^* = (\lambda^*_B: B \in \mathcal{B})$ is any optimal $\lambda \in \Lambda$ for the linear program in \eqref{skcapacity}. 

\begin{lemma}
For any function $\textbf{L}$ of $X_{\mathcal{M}}$, we have 
$$
n\textbf{I}(X_{\mathcal{M}}) =  \textbf{I}(X_{\mathcal{M}}^n|\textbf{L})+H(\textbf{L})-\sum_{B\in \mathcal{B}} \lambda_B^*H(\textbf{L}|X_{B^c}^n). 
$$
\label{lemma:eq}
\end{lemma}

\begin{IEEEproof}
Consider $\textbf{L}=\textbf{L}(X_{\mathcal{M}}^n)$. From \eqref{skcapacity}, we have
\begin{align*}
n\textbf{I}(X_{\mathcal{M}}) & = H(X_{\mathcal{M}}^n)-\sum_{B\in \mathcal{B}} \lambda_B^* H(X_B^n|X_{B^c}^n) \nonumber \\
                                                & = H(X_{\mathcal{M}}^n,\textbf{L})-\sum_{B\in \mathcal{B}} \lambda_B^* H(X_B^n,\textbf{L}|X_{B^c}^n) \\
                                                & = H(X_{\mathcal{M}}^n|\textbf{L}) + H(\textbf{L}) \\
                                                &\: \:\:\:\:\: -\sum_{B\in \mathcal{B}} \lambda_B^*H(X_B^n|\textbf{L},X_{B^c}^n) -\sum_{B\in \mathcal{B}}\lambda_B^*H(\textbf{L}|X_{B^c}^n) \\
                                                & = \textbf{I}(X_{\mathcal{M}}^n|\textbf{L})+H(\textbf{L})-\sum_{B\in \mathcal{B}} \lambda_B^*H(\textbf{L}|X_{B^c}^n) 
\end{align*}
the last equality above being due to \eqref{cmi}.
\end{IEEEproof}

\begin{lemma}
For any CR $\textbf{J}$ obtained from an interactive communication $\textbf{F}$,
$$
\lim_{n \to \infty} \frac{1}{n}\sum_{B\in \mathcal{B}} \lambda_B^* H(\textbf{J}|X_{B^c}^n,\textbf{F}) = 0.
$$
\label{lem:JF}
\end{lemma}
\begin{IEEEproof}
Fix an $\epsilon > 0$. We have for all sufficiently large $n$, by Fano's inequality,
\begin{align}
\frac{1}{n}\sum_{B\in \mathcal{B}} \lambda_B^* H(\textbf{J}|X_{B^c}^n,\textbf{F})
           & \leq \frac{1}{n}\sum_{B\in \mathcal{B}}\lambda_B^* \left( h(\epsilon)+\epsilon H(X_{B^c}^n,\textbf{F})\right) \nonumber \\
           &  \leq \frac{1}{n}\sum_{B\in \mathcal{B}}\lambda_B^* \left( h(\epsilon)+\epsilon H(X_{\mathcal{M}}^n,\textbf{F})\right) \nonumber \\
           &  = \frac{1}{n}\sum_{B\in \mathcal{B}}\lambda_B^* \left( h(\epsilon)+\epsilon H(X_{\mathcal{M}}^n)\right) \nonumber \\
           &  = \frac{1}{n}\sum_{B\in \mathcal{B}}\lambda_B^* \left( h(\epsilon)+n \epsilon H(X_{\mathcal{M}})\right) \nonumber \\
           &  \leq (2^m-2)\left[h(\epsilon) + \epsilon  H(X_{\mathcal{M}}) \right] \label{th:main:a} 
\end{align}
where $h(.)$ is the binary entropy function, and (\ref{th:main:a}) follows from the fact that, by definition, $\lambda_B^*\leq 1$ and $|\mathcal{B}|=2^m-2$. Note that the expression in \eqref{th:main:a} goes to $0$ with $\epsilon$, since $h(\epsilon)\to 0$ as $\epsilon \to 0$, and $H(X_{\mathcal{M}}) \le \log(\prod_{j=1}^m |\mathcal{X}_j|)$.
\end{IEEEproof}

The last lemma we need, stated without proof, is a special case of \cite[Lemma B.1]{CN08}.

\begin{lemma}[\cite{CN08}, Lemma B.1]
For an interactive communication $\textbf{F}$ we have
$$
H(\textbf{F})\geq \sum_{B\in \mathcal{B}} \lambda_B^*H(\textbf{F}|X_{B^c}^n).
$$
\label{lem:comm}
\end{lemma}

\vspace*{-15pt}
With these lemmas in hand, we can proceed to the proof of Theorem~\ref{th:main}.

\begin{IEEEproof}[Proof of Theorem \ref{th:main}]
The proof is done in two parts. In the first part, we prove that $R_{\CI}^r\geq \CI^r(X_{\mathcal{M}})-\textbf{I}(X_{\mathcal{M}})$. In the second part, we show that $R_{\SK}^r\geq R_{\CI}^r$.

\textit{Part~I: $R_{\CI}^r \geq \CI^r(X_{\mathcal{M}})-\textbf{I}(X_{\mathcal{M}})$}

The idea is to show that $\textbf{I}(X_{\mathcal{M}})+R_{\CI}^r$ is an achievable $\CI^r$ rate, so that $\CI^r(X_{\mathcal{M}}) \le \textbf{I}(X_{\mathcal{M}})+R_{\CI}^r$.

Fix an $\epsilon > 0$. By the definition of $R_{\CI}^r$, for all sufficiently large $n$, there exists an $r$-interactive communication $\textbf{F}$ satisfying $\frac{1}{n}\log |\cF| \; \leq R_{\CI}^r+\epsilon/2$ and a CR $\textbf{J}$ such that $\textbf{L}=(\textbf{J},\textbf{F})$ is a $\CI^r$. We will show that $\frac{1}{n} H(\textbf{J},\textbf{F}) \le \textbf{I}(X_{\mathcal{M}})+R_{\CI}^r+\epsilon$ for all sufficiently large $n$. This, by Definition~\ref{def:CIrate}, suffices to show that $\textbf{I}(X_{\mathcal{M}})+R_{\CI}^r$ is an achievable $\CI^r$ rate.

Setting $\textbf{L} = (\textbf{J},\textbf{F})$ in Lemma~\ref{lemma:eq}, we obtain
\begin{align}
& \frac{1}{n}\left[H(\textbf{J},\textbf{F})  -\sum_{B\in \mathcal{B}} \lambda_B^* H(\textbf{F}|X_{B^c}^n)\right] -\textbf{I}(X_{\mathcal{M}}) \nonumber \\
& \hspace*{3em}
= \ \frac{1}{n}\left[\sum_{B\in \mathcal{B}} \lambda_B^* H(\textbf{J}|X_{B^c}^n,\textbf{F})-\textbf{I}(X_{\mathcal{M}}^n|\textbf{J},\textbf{F})\right] \nonumber \\
& \hspace*{3em}
\leq \ \epsilon/2, \label{th:main:c}
\end{align}
where (\ref{th:main:c}) follows from Lemma~\ref{lem:JF}. Re-arranging, we get
\begin{align*}
\frac{1}{n} H(\textbf{J},\textbf{F}) & \le \textbf{I}(X_{\mathcal{M}})+\frac{1}{n} \sum_{B\in \mathcal{B}} \lambda_B^* H(\textbf{F}|X_{B^c}^n) +\epsilon/2 \\
& \le \textbf{I}(X_{\mathcal{M}})+\frac{1}{n} H(\textbf{F})+\epsilon/2
\end{align*}
the second inequality coming from Lemma~\ref{lem:comm}. Finally, using the fact that $\frac{1}{n}H(\textbf{F})\leq \frac{1}{n}\log|\cF| \; \leq R_{\CI}^r+\epsilon/2$, we see that
$$
\frac{1}{n} H(\textbf{J},\textbf{F}) \leq \textbf{I}(X_{\mathcal{M}})+R_{\CI}^r+\epsilon 
$$
which is what we set out to prove.

\textit{Part II: $R_{\SK}^r\geq R_{\CI}^r$}

Fix $\epsilon > 0$. From the definition of $R_{SK}^r$, there exist an $r$-interactive communication $\textbf{F}$ and an SK $\textbf{K}$ obtained from $\textbf{F}$ such that, for all sufficiently large $n$, $\frac{1}{n}\log|\cF| \; \leq R_{SK}^r+\epsilon$ and $\frac{1}{n}H(\textbf{K})\geq \textbf{I}(X_{\mathcal{M}})-\epsilon$. We wish to show that $(\textbf{K},\textbf{F})$ is a $\CI^r$, so that by Definition~\ref{def:RCI}, we would have  $R_{\SK}^r\geq R_{\CI}^r$.

Setting $\textbf{L}=(\textbf{K},\textbf{F})$ in Lemma~\ref{lemma:eq}, we have for all sufficiently large $n$,
\begin{align}
\frac{1}{n}\textbf{I}(X_{\mathcal{M}}^n|\textbf{K},\textbf{F}) 
            & = \textbf{I}(X_{\mathcal{M}})-\frac{1}{n}H(\textbf{K},\textbf{F})+\frac{1}{n}\sum_{B\in \mathcal{B}} \lambda_B^*H(\textbf{F}|X_{B^c}^n) \nonumber \\
            & \hspace{0.47cm}+\frac{1}{n}\sum_{B\in \mathcal{B}}\lambda_B^*H(\textbf{K}|X_{B^c}^n,\textbf{F}) \nonumber \\
            & \leq  \textbf{I}(X_{\mathcal{M}})-\frac{1}{n}H(\textbf{K}|\textbf{F}) + \epsilon \label{th:main:e} \\
            & \leq   \textbf{I}(X_{\mathcal{M}})-\frac{1}{n}H(\textbf{K})+\epsilon+\epsilon \label{th:main:f} \\ 
            &  \leq 3\epsilon, \label{th:main:g}
\end{align}
where (\ref{th:main:e}) follows from Lemmas~\ref{lem:JF} and \ref{lem:comm}, (\ref{th:main:f}) follows from the fact that $I(\textbf{K};\textbf{F})\leq \epsilon$, while (\ref{th:main:g}) is due to the fact that $\frac{1}{n}H(\textbf{K})\geq \textbf{I}(X_{\mathcal{M}})-\epsilon$. Thus, by Definition~\ref{def:CI}, $(\textbf{K},\textbf{F})$ is a $\CI^r$.
\end{IEEEproof}

We do not know if the lower bounds of Theorem~\ref{th:main} are in general tight, in the sense of there being matching upper bounds. For the special case of the two-terminal model, Theorem~3 of \cite{Tyagi13} shows that the bound on $R_{\SK}$ is tight (albeit under a weaker notion of SK, as explained in Footnote~\ref{fn:weakSK}). Another issue with our Theorem~\ref{th:main} is that the bounds are difficult to evaluate explicitly, as we do not have a computable characterization of $\CI^r(X_{\mathcal{M}})$ or $\CI(X_{\mathcal{M}})$. However, in the next section, we show that a version of our bound can be computed exactly in the case of the complete graph PIN model, where it matches an upper bound known from \cite{NN10}.

\section{The Linear Communication Complexity of the Complete Graph PIN Model\label{sec:linear}}

Throughout this section, we focus solely on the PIN model of Nitinawarat and Narayan \cite{NN10}, which we quickly review first. The model is defined on an underlying graph $\cG = (\cV,\cE)$ with $\cV = \cM$, the set of $m$ terminals of the model. For $n \in \N$, define $\cG^{(n)}$ to be the multigraph $(\cV,\cE^{(n)})$, where $\cE^{(n)}$ is the multiset of edges formed by taking $n$ copies of each edge of $\cG$. Associated with each edge $e \in \cE^{(n)}$ is a Bernoulli$(1/2)$ random variable $\xi_e$; the $\xi_e$s associated with distinct edges in $\cE^{(n)}$ are independent. With this, the random variables $X_i^n$, for $i\in\mathcal{M}$, are defined as $X_i^n=(\xi_e$ : $e\in\mathcal{E}^{(n)}$ and $e$ is incident on $i$). When $\cG = K_m$, the complete graph on $m$ vertices, we have the \emph{complete graph PIN model}. 

The SK capacity, $\textbf{I}(X_{\cM})$, of a PIN model defined on a graph $\cG$ is equal to the ``spanning tree packing rate'' of $\cG$ \cite[Theorem~5]{NN10}. When $\cG = K_m$, this can be computed to be $m/2$ \cite{TKSV12}. As mentioned in the Introduction (see Footnote~\ref{fn:linSK}), it is known that in the PIN model, a maximal-rate SK can be generated by a protocol in which the public communication $\textbf{F}$ and the SK $\textbf{K}$ are \emph{linear} functions of $X_{\cM}^n$. Of course, to have linear functions, we must assume that all the underlying alphabets, $\cX_i$, $\cF$ etc., are linear spaces --- indeed, we take them to be finite-dimensional vector spaces over the binary field $\F_2$. As shown in \cite{NN10}, a maximal-rate SK $\textbf{K}$ (which may be taken to be a linear function of $X_{\cM}^n$) can be obtained from an omniscience-enabling linear $m$-interactive communication $\textbf{F}$ of rate $R_{\CO} \triangleq H(X_\cM) - \textbf{I}(X_\cM)$. The quantity $R_{\CO}$ is the minimum rate of communication (not necessarily linear) that enables omniscience at all terminals.

It is natural to ask whether a lower rate of communication could suffice to achieve SK capacity within the PIN model, when the communication and the SK are restricted to be linear functions of $X_{\cM}^n$. To formulate this question precisely, we modify Definition~\ref{def:RSKr} by additionally requiring that the $r$-interactive communication $\textbf{F}$ and the SK $\textbf{K}$ be linear functions of $X_{\cM}^n$, and denote by $LR_{\SK}^r$ the infimum over all achievable rates as per the modified definition. Analogous to \eqref{def:RSK}, we define the \emph{linear communication complexity} of generating a maximal-rate SK to be $LR_{\SK} = \lim_{r\to\infty} LR_{\SK}^r$. By the discussion before the definition of $R_{\CO}$ above, we obviously have $LR_{\SK} \le LR_{\SK}^r \le R_{\CO}$ for all $r \ge m$. The question is whether $LR_{\SK} = R_{\CO}$. 

To answer this question, we need lower bounds on $LR_{\SK}^r$ and $LR_{\SK}$. Bounds analogous to those in Theorem~\ref{th:main} can be readily obtained by simply modifying the appropriate definitions. Thus, for any PIN model, we define $\LCI_W$, $\LCI^r$, $\LCI_W(X_{\cM})$, $\LCI^r(X_{\cM})$, $LR_{\CI}^r$ and $LR_{\CI}$ analogous to $\CI_W$, $\CI^r$, $\CI_W(X_{\cM})$, $CI^r(X_{\cM})$, $R_{\CI}^r$ and $R_{\CI}$, respectively, by modifying Definitions~\ref{def:CI}--\ref{def:RCI} to include the additional requirement that $\textbf{L}$, $\textbf{J}$ and $\textbf{F}$ be linear functions of $X_{\cM}^n$. The arguments of Section~\ref{sec:main} then show that the linear analogues of Proposition~\ref{prop:ineqs} and Theorem~\ref{th:main} hold for any PIN model. For future reference, we record two inequalities in particular: for all $r \ge m$,
\begin{gather}
H(X_{\mathcal{M}}) \ge \LCI^r(X_{\mathcal{M}}) \ge \LCI_W(X_{\mathcal{M}}) \label{LCI_ineq1} \\
LR_{\SK}^r \geq \LCI^r(X_{\mathcal{M}})-\textbf{I}(X_{\mathcal{M}}) \label{LCI_ineq2}
\end{gather}

From this point on, we restrict our attention to the complete graph PIN model, where we are actually able to determine $LR_{\SK}$ exactly. For this model, we know that $\textbf{I}(X_{\cM}) = m/2$ \cite{TKSV12}, and hence, $R_{\CO} = \binom{m}{2} - m/2 = m(m-2)/2$. Thus, we have
\begin{equation}
LR_{\SK} \le LR_{\SK}^r \le m(m-2)/2
\label{LR_pin_ineq}
\end{equation}
for all $r \ge m$. The theorem below states that the inequalites in~\eqref{LR_pin_ineq} are all equalities.

\begin{theorem}
For the PIN model defined on the complete graph $K_m$, we have for any $r \ge m$,
$$
LR_{\SK}^r=LR_{\SK}=R_{\CO} = m(m-2)/2.
$$
\label{th:CPIN}
\end{theorem}

Thus, for the complete graph PIN model, we are able to answer in the affirmative the question asked earlier of whether $LR_{\SK} = R_{\CO}$. In particular, the linear communication complexity of this PIN model is achieved by the SK generation protocol of \cite{NN10} that goes through omniscience at all terminals.

The remainder of this section is devoted to a proof of Theorem~\ref{th:CPIN}. The key idea is to compute $\LCI^r(X_{\cM})$, for which we need a means of dealing with the quantity $\textbf{I}(X_{\cM}^n|\textbf{L})$ when $\textbf{L}$ is a linear function of $X_{\cM}^n$. To start with, we explicitly determine $\lambda^*$, the optimal $\lambda \in \Lambda$ for the linear program in \eqref{skcapacity}. If we define $\tilde{\lambda} = (\tilde{\lambda}_B: B \in \cB)$ such that $\tilde{\lambda}_B = \frac{1}{m-1}$ whenever $|B| = m-1$, and $\tilde{\lambda}_B = 0$ otherwise, then it can be easily verified that $\tilde{\lambda} \in \Lambda$, and moreover,
$$
H(X_{\mathcal{M}})- \sum_{B \in \mathcal{B}} \tilde{\lambda}_B H(X_B|X_{B^c}) = m/2.
$$
Since we know that $\textbf{I}(X_{\cM}) = m/2$, we infer from \eqref{skcapacity} that $\tilde{\lambda}$ is an optimal $\lambda \in \Lambda$, i.e., $\lambda^* = \tilde{\lambda}$. Hence, for the complete graph PIN model, \eqref{cmi} reduces to 
\begin{equation}
\textbf{I}(X_{\mathcal{M}}^n | \textbf{L}) = H(X_{\mathcal{M}}^n | \textbf{L})- \frac{1}{m-1} \sum_{i=1}^m H(X_{\cM \setminus \{i\}}^n|X_{i}^n,\textbf{L})  \label{cmi_pin}
\end{equation}
Now consider any linear function $\textbf{L}$ of $X_{\mathcal{M}}^n$. The fact that $\textbf{L}$ is a function of $X_{\cM}^n$ allows us to further simplify \eqref{cmi_pin}:
\begin{align}
\textbf{I}(X_{\mathcal{M}}^n | \textbf{L}) 
   & = H(X_{\mathcal{M}}^n) - H(\textbf{L}) \notag \\
   & \hspace{1.2em} - \frac{1}{m-1} \sum_{i=1}^m \left[H(X_{\cM}^n) - H(X_i^n) - H(\textbf{L}|X_i^n)\right] \notag \\
   & = \frac{nm}{2} - H(\textbf{L}) + \frac{1}{m-1} \sum_{i=1}^m H(\textbf{L}|X_i^n) \label{cmi_pin2},
\end{align}
the equality \eqref{cmi_pin2} using the facts that $H(X_{\cM}^n) = n\binom{m}{2}$ and $H(X_i^n) = n(m-1)$. 

To proceed further, we need to use the linearity of $\textbf{L}$. Observe that $\textbf{L}(X_{\cM}^n)$ can be viewed as the product $L \xi$ over the binary field $\F_2$, where $\xi$ is the random vector $(\xi_e\: : \: e\in\mathcal{E}^{(n)})$ and $L$ is a (deterministic) matrix over $\F_2$ with $\frac{nm(m-1)}{2}$ columns. The columns of $L$ are indexed by the set $\cE^{(n)}$; the indexing of columns of $L$ is in the same order as the indexing of the coordinates of $\xi$. For $i \in \cM$, let $\cE_i = \{e \in \cE^{(n)}: e \text{ is incident with } i\}$. 

The lemma below allows us to express $H(\textbf{L})$ and $H(\textbf{L}|X_i^n)$ in \eqref{cmi_pin2} in terms of the ranks of certain submatrices of $L$. 

\begin{lemma}
Let $Y=(Y_1,Y_2,\cdots,Y_p)$ be a vector of i.i.d.\ Bernoulli$(1/2)$ random variables and $A$ be any matrix over $\F_2$ with $p$ columns. Consider $Z=AY$, all operations being over $\F_2$. For $S \subseteq\{1,2,\cdots,p\}$, let $Y_S = (Y_i: i \in S)$, and let $A|_S$ denote the submatrix of $A$ consisting of the columns indexed by $S$. We then have
\begin{itemize}
\item[(a)] $H(Z)=\text{rank}(A)$.
\item[(b)] $H(Z|Y_S)=\text{rank}(A|_{S^c})$.
\end{itemize}
\label{lem:rank}
\end{lemma}

We defer the proof of the lemma till the end of this section. Returning to \eqref{cmi_pin2}, Lemma~\ref{lem:rank} shows that $H(\textbf{L}) = \rank(L)$, and $H(\textbf{L}|X_i^n) = \rank(L|_{\cE_i^c})$, where 
$$
\cE_i^c = \{e \in \cE^{(n)}: e \text{ is not incident with } i\}.
$$
Thus, \eqref{cmi_pin2} becomes
\begin{equation}
\textbf{I}(X_{\mathcal{M}}^n | \textbf{L}) = \frac{nm}{2} - \rank(L) + \frac{1}{m-1} \sum_{i=1}^m \rank(L|_{\cE_i^c}). \label{cmi_pin3}
\end{equation}

As the final step in our processing of $\textbf{I}(X_{\mathcal{M}}^n | \textbf{L})$, we derive a lower bound on the last term of \eqref{cmi_pin3}. Let $t = \rank(L)$, and let $T = \{e_1,e_2,\ldots,e_t\} \subseteq \cE^{(n)}$ be a subset of the columns of $L$ that form a basis for its column space. We then have
\begin{align}
\sum_{i=1}^m \text{rank}(L|_{\cE_i^c}) 
                 & \geq \sum_{i=1}^m |T \cap \cE_i^c| \nonumber \\
                 & = \sum_{i=1}^m \sum_{\ell=1}^t \mathbb{I}_{\cE_i^c}(e_\ell) \nonumber \\
                 & = \sum_{\ell=1}^t \sum_{i=1}^m \mathbb{I}_{\cE_i^c}(e_\ell) \nonumber \\
                 & = \sum_{\ell=1}^t (m-2) \label{eval:c} \\
                 & = (m-2) \text{rank}(L) \label{eval:d}
\end{align}
where $\mathbb{I}_{\cE_i^c}(e_\ell)$ equals $1$ if $e_\ell\in \cE_i^c$, and equals $0$ otherwise. The equality in \eqref{eval:c} is due to the fact that any edge $e_{\ell}$ is incident on exactly two vertices, and hence, is \emph{not} incident on exactly $m-2$ vertices $i \in \cM$. Plugging \eqref{eval:d} back into \eqref{cmi_pin3}, we obtain
\begin{equation}
\textbf{I}(X_{\mathcal{M}}^n | \textbf{L}) \ge \frac{nm}{2} - \frac{1}{m-1} \rank(L). \label{cmi_pin4}
\end{equation}

We are now in a position to compute $\LCI^r(X_{\cM})$ using \eqref{LCI_ineq1}. The upper bound gives us $\LCI^r(X_{\cM}) \le \binom{m}{2}$ for all $r \ge m$. For the lower bound, let $\textbf{L}$ be any $\LCI_W$ so that (by the linear analogue of) Definition~\ref{def:CI}, for any $\epsilon > 0$, we have $\frac{1}{n} \textbf{I}(X_{\mathcal{M}}^n | \textbf{L}) \le \frac{\epsilon}{m-1}$ for all sufficiently large $n$. The bound in \eqref{cmi_pin4} now yields $\frac{m}{2} - \frac{1}{n(m-1)} \rank(L) \le \frac{\epsilon}{m-1}$, or equivalently, $\frac{1}{n} \rank(L) \ge \binom{m}{2} - \epsilon$ for all sufficiently large $n$. Thus, for any $\epsilon> 0$, we have $\frac{1}{n} H(\textbf{L}) \ge \binom{m}{2} - \epsilon$ for all sufficiently large $n$. Hence, from Definition~\ref{def:CIrate}, it follows that $\LCI_W(X_{\cM}) \ge \binom{m}{2}$. Consequently, from the upper and lower bounds in \eqref{LCI_ineq1}, we obtain $\LCI^r(X_{\cM}) = \binom{m}{2}$ for all $r \ge m$.  

From \eqref{LCI_ineq2}, we now have
$
LR_{\SK}^r \ge \binom{m}{2} - m/2 = m(m-2)/2
$
for all $r \ge m$. Together with \eqref{LR_pin_ineq}, this yields $LR_{\SK}^r = m(m-2)/2$ for all $r \ge m$, and hence, $LR_{\SK} = m(m-2)/2$ as well. This completes the proof of Theorem~\ref{th:CPIN}, modulo the proof of Lemma~\ref{lem:rank}, which we give below.

\emph{Proof of Lemma~\ref{lem:rank}}\footnote{The authors would like to thank Shashank Vatedka for the proof of part~(b) of the lemma.}: Part~(a) follows immediately from \cite[Theorem 7.3]{NCT06}. For part~(b), we introduce some notation: for $S \subseteq \{1,2,\ldots,p\}$, let $\tilde{Y}_S = (\tilde{Y}_1,\ldots,\tilde{Y}_p)$ be defined by setting $\tilde{Y}_i = Y_i$ if $i \in S$, and $\tilde{Y}_i = 0$ otherwise. Then,
\begin{align*}
H(Z | Y_S) & = H(Z + A\tilde{Y}_S \mid Y_S) \\
& = H(A(Y+\tilde{Y}_S) \mid Y_S) \\
& = H(A\tilde{Y}_{S^c} \mid Y_S) \\
& = H(A|_{S^c} Y_{S^c} \mid Y_S) \\
& = H(A|_{S^c} Y_{S^c}) \\
& = \rank(A|_{S^c})
\end{align*}
by part~(a) of the lemma.
\endproof

\section{Concluding Remarks} \label{sec:conc}

We regard the work presented in this paper as the first step towards characterizing a rate region for the communication needed to generate a maximal-rate SK. We have given lower bounds on the \emph{total sum rate}, i.e., the sum of the rates of communication from all terminals. A next step would be to find bounds on \emph{partial sum rates}, i.e., the sum of the rates of communication from a subset of the terminals in $\cM$. We expect that such bounds will be needed to characterize the communication rate region, analogous to that in the distributed source coding (Slepian-Wolf) problem of information theory. 

Another important open problem is to find computable characterizations of the rates $\CI^r(X_{\mathcal{M}})$ and $\CI_W(X_{\cM})$. At the very least, we would like to be able to explicitly evaluate these rates in some special cases, such as in the PIN model. We expect that the linear setting results of Section~\ref{sec:linear} should be easily extendable to a wider class of PIN models. 

Finally, we remark that the bounding technique used in Theorem~\ref{th:main} can also be used to get bounds on the minimum rate of communication needed to generate maximal-rate private keys (as defined in \cite{CN04}) and maximal-rate keys when some terminals are silent (as defined in \cite{GA10}).

\end{document}